\documentclass[aps, preprint, groupedaddress, showpacs, floatfix]{revtex4}

\usepackage{amsmath, amssymb}

\usepackage[sort&compress]{natbib}
\usepackage{amsthm}
\usepackage{times}

\usepackage[english]{babel}
\usepackage{xy}
\usepackage{graphicx}
\usepackage{epsfig}
\usepackage{subfigure}

\newtheorem{thm}{Theorem}[section]

\newtheorem{lem}[thm]{Lemma}
\newtheorem{prop}[thm]{Proposition}

\newtheorem{defn}[thm]{Definition}
\newtheorem{alg}[thm]{Algorithm}

\newtheorem{rem}[thm]{Remark}

\newcommand{\pquad}{\qquad}

\newcommand{\C}{\mathbb{C}}
\newcommand{\Cn}{\C^n}

\newcommand{\Z}{\mathbb{Z}}

\newcommand{\N}{\mathbb{N}}

\newcommand{\M}{\mathcal{M}}

\newcommand{\Rr}{\mathbb{R}}

\newcommand{\Rn}{\Rr^n}

\newcommand{\Cc}{\mathcal{C}}
\newcommand{\F}{\mathcal{F}}

\newcommand{\Cz}{\C[z]}

\DeclareMathOperator{\FP}{Fix}

\DeclareMathOperator{\sFP}{sFix}

\DeclareMathOperator{\id}{id}
\DeclareMathOperator{\card}{card}
\DeclareMathOperator{\arccot}{arccot}

\newcommand{\Fref}[1]{Figure~\ref{#1}}
\newcommand{\Eref}[1]{Equation~\eqref{#1}}
\newcommand{\eref}[1]{\eqref{#1}}
\newcommand{\Sref}[1]{\S\ref{#1}}
\newcommand{\e}{\varepsilon}

\renewcommand{\e}{\varepsilon}
\newcommand{\ud}{\mathrm{d}}
\newcommand{\ite}[2]{#1^{\circ #2}}
\newcommand{\vphi}{\varphi}

\newcommand{\Ch}{{\hat{\C}}}

\newcommand{\abs}[1]{\left|#1\right|}
\newcommand{\dist}[2]{\left\|#1 - #2\right\|}
\newcommand{\lda}[2]{\rho_{#1}(#2)}
\newcommand{\sset}[1]{\left\lbrace #1\right\rbrace}
\newcommand{\rset}[2]{\left\lbrace #1\ \left|\ #2\right.\right\rbrace}

\newcommand{\set}[2]{\rset{#1}{#2}}

\newcommand{\ind}[1]{_{\mathrm{#1}}}

\newcommand{\alts}[1]{\tilde{#1}}

\newcommand{\xf}{{x^{*}}}
\newcommand{\zf}{{z^{*}}}

\newcommand{\norm}[1]{\|#1\|}
\newcommand{\enorm}{\norm{\,\cdot\,}}
\newcommand{\ball}[2]{B(#1, #2)}
\newcommand{\cball}[2]{\overline{\ball{#1}{#2}}}
\newcommand{\drv}[2]{\ud#1|_{#2}}

\begin{document}

\title{Adapting Predictive Feedback Chaos Control for Optimal Convergence Speed}
\author{Christian Bick${}^{1,2}$}
\author{Marc Timme${}^{1,3}$}
\author{Christoph Kolodziejski${}^{1,4}$}
\address{${}ö1$ Network Dynamics Group, Max Planck Institute for Dynamics and Self-Organization (MPIDS), 37077 G\"ottingen, Germany and Bernstein Center for Computational Neuroscience, G\"ottingen\\
${}ö2$ Institute for Mathematics, Georg--August--Universit\"at G\"ottingen, 37073 G\"ottingen, Germany\\
${}ö3$ Institute for Nonlinear Dynamics, Georg--August--Universit\"at G\"ottingen, 37077~G\"ottingen, Germany\\
${}ö4$ Third Institute of Physics --- Biophysics, Georg--August--Universit\"at G\"ottingen, 37077~G\"ottingen, Germany}
\email{bick@nld.ds.mpg.de}

\begin{abstract}
Stabilizing unstable periodic orbits in a chaotic invariant set not
only reveals information about its structure but also leads to various
interesting applications. For the successful application of a chaos
control scheme, convergence speed is of crucial importance. Here we
present a predictive feedback chaos control method that adapts a
control parameter online to yield optimal asymptotic convergence
speed. We study the adaptive control map both analytically and
numerically and prove that it converges at least linearly to a value
determined by the spectral radius of the control map at the periodic
orbit to be stabilized. The method is easy to implement
algorithmically and may find applications for adaptive online control
of biological and engineering systems. 
\end{abstract}

\maketitle



\section{Introduction}

Some chaotic attractors contain infinitely many unstable periodic
orbits. These can be seen as a ``skeleton'' for the chaotic attractor,
therefore revealing important information about the dynamics of the
system itself. By suitable perturbations the stability of these unstable
periodic points can be changed; a control perturbation
renders them stable. Such ``chaos control'' has applications in many fields
\cite{Scholl2007}, including biological \cite{Rabinovich1998} and artificial
neural networks \cite{Steingrube2010, Scholl2010}.

In the last twenty years, different methods for stabilizing unstable
periodic orbits have been suggested. The seminal work by Ott, Grebogi, and
Yorke (OGY) \cite{Ott1990} and its implementations employ arbitrary small
perturbations of a parameter of the system to stabilize a known unstable
periodic orbit of a discrete time dynamical system. A successful
application of the OGY method, however, requires prior knowledge about or
online analysis of the dynamics to determine fixed points and their stability
properties.

A different approach is given by predictive feedback control (PFC) 
\cite{DeSousaVieira1996, Polyak2005} which overcomes this disadvantage.
In this approach the future state of the dynamics calculated from the
current state is fed back into the system to stabilize a periodic orbit.
This feedback control is noninvasive, i.e., the control strength vanishes
upon convergence, and is extremely easy to implement. It is a special
case of a recent effort to stabilize all periodic
points of a discrete time dynamical system \cite{Schmelcher1997, Schmelcher1998}
which is also closely related to nonlinear successive over-relaxation
methods \cite{Yang2010, Brewster1984}. It has been extensively studied
\cite{Diakonos1998, Pingel2000, Pingel2004, Crofts2006, Crofts2009} and
extended \cite{Doyon2002, Davidchack1999} with respect to its original
purpose as a tool for examining the structure of chaotic attractors.

In any real world application, speed of convergence is of crucial
importance. For example, if a robot is controlled by stabilizing
periodic orbits in a chaotic attractor~\cite{Steingrube2010}, the
time it needs to react to a changing environment is bounded by the
time the system needs to converge to a periodic point of a given
period. Hence, in praxis, one desires to tune the control parameter such that
the spectral radius of the unknown periodic point which the system
converges to is minimized. To the best of our knowledge, previous
works on chaos control have not considered convergence speed
while maintaining its simplicity in terms of implementation. Adaptation
of the control parameter has an impact on convergence speed. However, 
existing adaptation mechanisms \cite{Steingrube2010, Lehnert2011} have
two major shortcomings; they do not optimize for speed and, for adaptation
of heuristic nature, may adapt the parameter to regimes where stabilization
fails.

Here we introduce an adaptation method that overcomes
these shortcomings. It adaptively tunes the control parameter online
to achieve optimal asymptotic
convergence speed. This work is organized as follows. In the second
section we review the PFC method and introduce the notation that
will be used throughout the paper. In the third section, we present
the adaptation method and prove its convergence properties. As an
example, the well-known logistic map is studied both analytically and
numerically in Sections~4 and~5 before giving some concluding
remarks.

\section{Preliminaries}

A differentiable map $f: \Rn \to \Rn$ gives rise to a dynamical system
through the evolution equation
\begin{equation}\label{eq:System}
x_{k+1} = f(x_k)
\end{equation}
with $x_k\in \Rn$ for all $k\in\Z$. The sequence $(x_k), k\in\N$ is called
an orbit of the dynamical system with initial condition $x_0$ and if
$\ite{f}{p}(x_k)=x_{k+p}$ for
all $k\geq 0$ we say that the orbit is periodic with period $p$. Here,
\[\ite{f}{p}=\underbrace{f\circ f \circ \cdots\circ f}_{p \text{ times}}\]
denotes the $p$-fold composition
of $f$. Let $\FP(f) := \set{x\in\Rn}{f(x)=x}$ be the set of fixed points,
i.e., periodic points of period one. Note that any periodic orbit is a
fixed point of the $p$-th iterate of the map~$f$ so we will use the
expressions fixed point and periodic orbit interchangeably.

Let $A\subset\Rn$ be a forward invariant subset of $\Rn$ with respect to $f$, i.e., 
$f(A)\subset A$. If periodic points are dense in $A$ and $f$ maps 
transitively, then we call $A$ a \emph{chaotic set}. Julia sets \cite{Milnor2006},
as described below,
are examples of such chaotic sets. Let $\drv{f}{x}$ denote the
derivative of $f$ at $x\in\Rn$ and $\id:\Rn\to\Rn$ the identity
map.

The results of \cite{Schmelcher1998} are now summarized as follows:

\begin{prop}\label{prop:SD}
Suppose $\xf\in \FP(f)$ and the derivative $\drv{f}{\xf}$ and $\drv{f}{\xf}-\id$
are real, nonsingular and diagonalizable. Then there exist parameters
$\mu > 0$ and an orthogonal matrix \mbox{$C_k\in O(n)$} such that $\xf$ is
an attractive fixed point of the map~$g_{\mu}$ obtained from $f$ through the
transformation
\[S(\mu, C_k): f \mapsto \id+\mu C_k(f-\id) =: g_{\mu}.\]
In particular, $S(\mu, C_k)$ preserves the set of fixed points, that is 
$\FP(f)=\FP(g_{\mu}).$
\end{prop}

In fact, it can be shown that the number of matrices $C_k$ needed to
stabilize all fixed points of a given map $f$ is quite limited. They
depend on the local stability properties of fixed points and there are
types of fixed points
that can be stabilized for $C_k\in\sset{\pm\id}$. For a given
$\xf\in\FP(f)$ with $\chi_j\in\C$, $j\in\sset{1,\ldots,n}$, being the
eigenvalues of $\ud g_{\mu}|_{\xf}$ we want to denote by
\[\lda{\xf}{\mu} := \max_{j\in\sset{1, \ldots, n}}\sset{\abs{\chi_j}}\]
the spectral radius, i.e., the maximum of the absolute values of the
eigenvalues of the derivative of $g_{\mu}$ at $\xf$. We have
\begin{equation}\label{eq:Derivative}
\ud g_{\mu}|_x=\id + \mu C_k(\ud f|_x-\id).
\end{equation}
for all $x\in\Rn$. In other words, the proposition above ensures the
existence of $\mu$ and $C_k(\xf)$ for a given $\xf\in\FP(f)$ such
that the transformation $S(\mu, C_k)$ gives $\lda{\xf}{\mu} < 1$,
cf. \Fref{fig:SrSD}. Therefore, with these parameters, the
fixed point $\xf$ of $f$ is an attracting fixed point for $g_{\mu}$.

\begin{figure} 
\begin{center}
\includegraphics[scale=1]{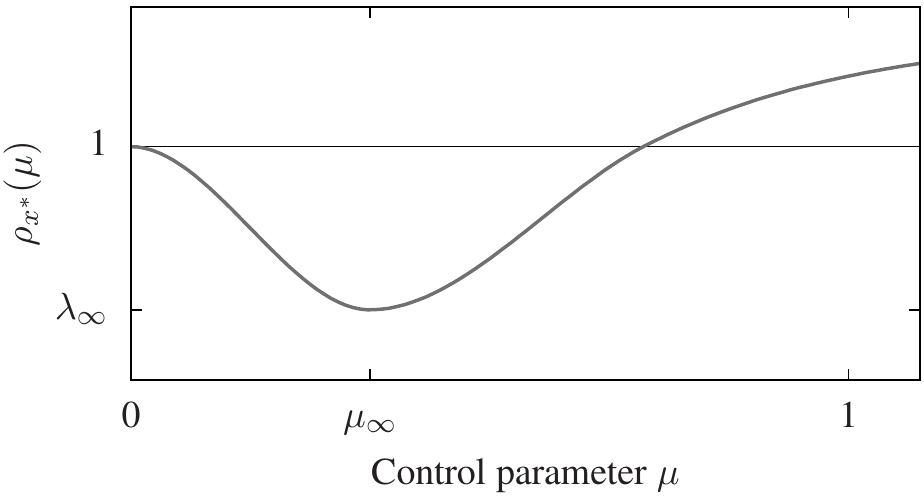}
\end{center}
\caption{\label{fig:SrSD}Sketch of the dependence of the spectral
radius $\lda{\xf}{\mu}$ on $\mu$ for some fixed point $\xf$ according
to Proposition~\ref{prop:SD}.}
\end{figure}

The results above are directly related to predictive feedback chaos
control methods. A transformation $T_\eta: f \mapsto g$ is called
a \emph{chaos control transformation} if $g$ can be written as 
$g=f+\eta c$ with \emph{control perturbation} $c: \Rn\to\Rn$ and
$\eta\in\Rr$. Note that in case $C_k \in \sset{\pm\id}$ the
transformations $S(C_k, \mu)$ are chaos control transformations
since
\[g_{\mu}=f + (1\mp\mu)(\id-f)\]
with $\eta=1-\mu$. Therefore, we will refer to these transformations
$S(\mu, C_k)$ as PFC transformations. Without loss of generality
we restrict ourselves to the case $C_k=\id$.

The results of Proposition~\ref{prop:SD}, however, give little information
about the speed of convergence, except for the fact that when decreasing
$\mu$ towards zero convergence takes longer and longer as the spectral
radius approaches one. In the vicinity of a stabilized fixed point,
convergence is at least linear and the
rate of convergence is bounded from above by the quantity $\lda{\xf}{\mu}$.
In order
to obtain an adaptation method that increases the speed of convergence
we therefore have to minimize $\lda{\xf}{\mu}$ using the control parameter
$\mu$. For a random initial condition, we do not know to what fixed
point $\xf$ (if any) the trajectory will converge to. We only have a
converging sequence $x_k\to \xf$. In other words, we are looking for a
way to obtain a sequence $\mu_k \to \mu_\infty$ where
\begin{equation}\label{eq:muinfty}
\mu_\infty = \sup\set{\mu > 0}{\forall \alts{\mu}>0:\lda{\xf}{\mu} \leq \lda{\xf}{\alts{\mu}}
\begin{array}{l}\text{and assumptions of}\\ \text{Prop.~\ref{prop:SD} are satisfied}\end{array}}
\end{equation}
the optimal $\mu$ to minimize $\lda{\xf}{\mu}$. Define
$\lambda_\infty := \lda{\xf}{\mu_\infty}$.

In applications, the control parameter $\mu$ plays a double role; on the
one hand, it can be used to turn chaos control on and off, $\mu=1$,
on the other hand it is the crucial parameter to stabilize the periodic
orbits and to determine the speed of convergence.

Define the class of functions
\[\F(\mu_0, p) := \set{f}{\card\left(\set{x\in \FP(\ite{f}{p})}{\lda{x}{\mu_0} < 1 \text{ for } C_k = \id}\right)>0},\]
with parameters $p\in\N$, and $\mu_0 > 0$ and where $\card$ denotes the
cardinality of a set. The sets $\F(\mu_0, p)$ are the functions $f$ with
a chaotic set that have at least one periodic orbit of period $p$ which
can be stabilized for the given parameters.


\section{An Adaptation Method to Accelerate Chaos Control}

In this section, suppose $f\in \F(\mu_0, p)$ for some $\mu_0 > 0$ and 
without loss of generality, $p=1$, since we can replace $f$ with the
$p$-th iterate. Suppose $g_{\mu}$ is the transformed map after
applying $S(\mu, \id)$. Furthermore we assume that for all times $k<0$
the system evolves according
to \Eref{eq:System}, i.e., with $\eta = 1-\mu = 0$, along a trajectory
of points in the chaotic set $A$. At time $k=0$ the control parameter~$\mu$
is set to~$\mu_0$. Therefore, because of the assumptions on $f$, there
is at least one periodic orbit of period $p$ on the chaotic attractor
which is now an attracting periodic orbit. Let $\sFP(f)$ denote the set
of these stabilized fixed points.

\subsection{Close to a fixed point}\label{sect:Close}

Recall two facts: any differentiable map $h\in \Cc^1(U)$ on an open set
$U\subset \Rn$ is Lipschitz-continuous on any compact $K\subset U$, that is
\[\dist{h(x)}{h(y)}\leq \norm{\ud h}_K\cdot\dist{x}{y},\]
for all $x,y\in K$ where $\enorm_K$ is the supremum of the operator norms
induced by a norm $\enorm$ on $K$ and $\ud h$ is the total derivative.
Furthermore, for any contraction $h$ on a Banach space $(X,\enorm)$,
i.e., a map  that satisfies
\[\dist{h(x)}{h(y)}\leq L\dist{x}{y}\]
with a Lipschitz constant $L<1$, the Banach Fixed Point Theorem gives
the existence of a unique fixed point $\xf$ together with the error estimates
$\dist{\xf}{x_k} \leq \frac{L^k}{1-L}\dist{x_0}{x_1}$ and
$\dist{x_{k+1}}{x_k} \leq L\dist{x_k}{x_{k-1}}$. Here, $x_k = \ite{h}{k}(x_0)$
for an initial condition $x_0\in X$.

Let $\xf\in\sFP(f)$ be fixed. According to Proposition~\ref{prop:SD}
there exists a $\lambda_0 < 1$ for $\mu = \mu_0$ sufficiently small such
that $\lda{\xf}{\mu_0} < \lambda_0$.
Therefore, there exists a vector norm $\|\ \cdot\ \|$ such that we
have $\norm{\drv{g_\mu}{\xf}}_{\textrm{op}}\leq\lambda_0$ for the induced
operator norm cf. \Fref{fig:Approximations}(a). We will omit the
index indicating the operator norm when it is clear from the context.
Henceforth all norms denote this vector norm and the induced operator
norm respectively.

\begin{figure}[h!] 
\begin{center}
\subfigure[]{\includegraphics[scale=1]{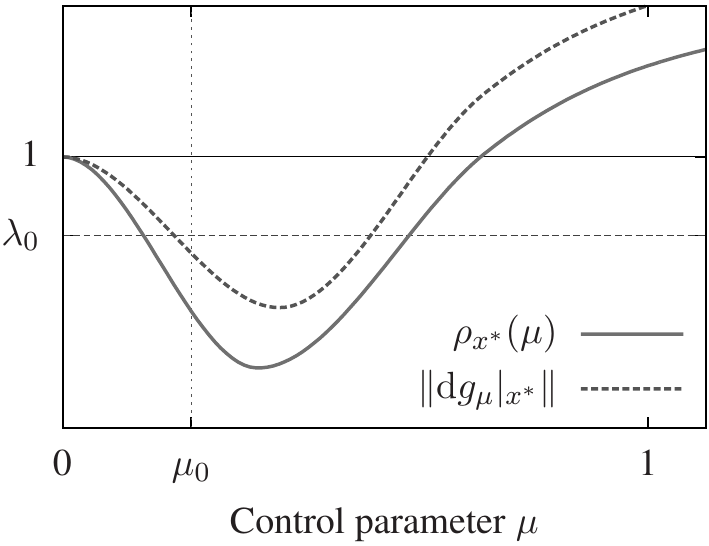}}\quad
\subfigure[]{\includegraphics[scale=1]{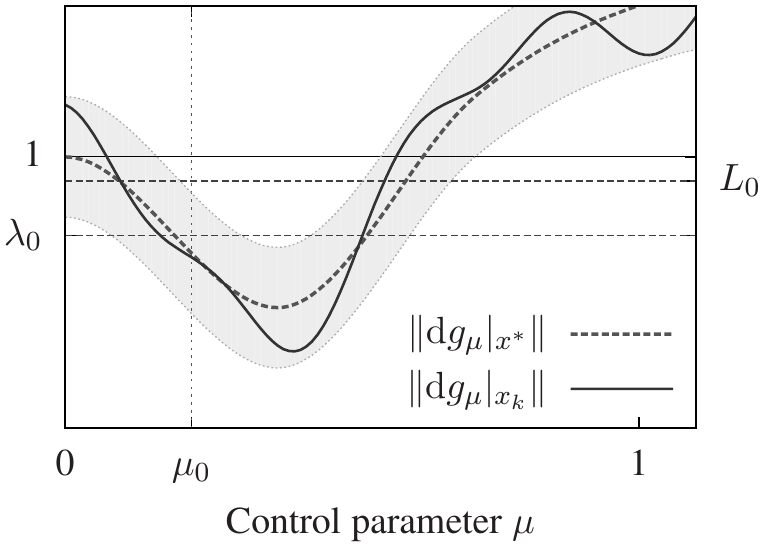}}
\end{center}
\caption{\label{fig:Approximations}(a) An appropriately chosen norm
approximates the spectral radius from above with
$\lda{\xf}{\mu_0}\leq \norm{\drv{g_{\mu_0}}{\xf}} \leq \lambda_0$.
(b) According to Condition~\eref{eq:Cond}, $\norm{\drv{g_{\mu}}{x_k}}$
lies in a $K\dist{\xf}{x_k}$-tube around $\norm{\drv{g_{\mu}}{\xf}}$.
While iterating, this tube becomes smaller and smaller. As another
consequence of~\eref{eq:Cond}, we have $\norm{\drv{g_{\mu}}{x}}_{\cball{\e}{\xf}}$
in a $K\e$-tube around $\norm{\drv{g_{\mu}}{\xf}}$. Therefore, for
$\e$ small enough, there is $L_0<1$ such that 
$\norm{\drv{g_{\mu_0}}{x}}_{\cball{\e}{\xf}}\leq L_0$.}
\end{figure}

Let $\ball{\e}{x}$ denote a ball of radius $\e$ centered at $x$ and
by $\cball{\e}{x}$ its closure. Assume that for $\alts{\e}$
small enough there is a constant $K \geq 0$ such that
\begin{equation}\label{eq:Cond}
\left|\|\ud g_{\mu}|_\xf\| - \|\ud g_{\mu}|_{{x}}\|\right| \leq K\dist{\xf}{{x}}
\end{equation}
for all $x$ with $\dist{\xf}{x} < \alts{\e}$ independent of $\mu$. This
condition is depicted in \Fref{fig:Approximations}(b). Now we can
choose $\e \leq \alts{\e}$ such that 
$\norm{\drv{g_{\mu_0}}{x}}_{\cball{\e}{\xf}} < 1$.
Put differently, for sufficiently small $\delta_0>0$ there exists
an $\e \in (0, \alts{\e})$ such that
$\norm{\drv{g_{\mu_0}}{x}}_{\cball{\e}{\xf}} \leq \lambda_0+\delta_0 =: L_0 < 1$.
The choice of~$\e$ (corresponding to the size of the ball around $\xf$)
depends on $\lambda_0$, $\mu_0$, and $\delta_0$.

\begin{rem}
Condition~\eref{eq:Cond} is satisfied in case $f$ has a bounded second
derivative on $\sFP(f)$ due to the functional dependence of the
derivative of $g_\mu$ on $\mu$ as given by~\eref{eq:Derivative}.
\end{rem}

\begin{alg}\label{alg:Opt}
For given $f\in \F(\mu_0, p), \mu_0, \lambda_0,$ and $K$ let 
$x_0\in\ball{\e}{\xf}$. The convergence acceleration algorithm consists of
the following steps:
\begin{description}
\item[\it Step 1] \emph{(Iterate)}: Calculate $x_1 = g_{\mu_0}(x_0)$.
\item[\it Step 2] \emph{(Optimize $\mu$)}: Minimize the ``cost function''
$\norm{\drv{g_\mu}{x_1}}$ with respect to $\mu\in (0, 1)$ under the conditions
\begin{subequations}
\begin{eqnarray}\label{eq:Opt}
&l(\mu):=\norm{\drv{g_\mu}{x_1}} + \left(\frac{2KL_0}{1-L_0}\right)\dist{x_0}{x_1} < L_0\\
&\mu \text{ maximal}
\end{eqnarray}
\end{subequations}
where $L_0 = \lambda_0 + \delta_0$ as above, cf. \Fref{fig:Optimization}.
\item[\it Step 3] \emph{(Set quantities)}: If the minimization under constraints of
Step~2 returned a result $\mu\ind{opt}$ then set $\mu_1 := \mu\ind{opt},
\lambda_1 := \norm{\drv{g_\mu}{x_1}} +
\left(\frac{KL_0}{1-L_0}\right)\dist{x_0}{x_1}$, 
$\delta_1 := \left(\frac{KL_0}{1-L_0}\right)\dist{x_0}{x_1}$ and
$L_1 := \lambda_1 + \delta_1$. Otherwise set $\mu_1 := \mu_0, \lambda_1 := \lambda_0,
\delta_1 := \delta_0$ and $L_1 := L_0$.
\end{description}
Repeat the steps with all indices increased by one.
\end{alg}

For this method we obtain the following results.

\begin{lem}\label{lem:Conv}
With the assumptions of Algorithm~\ref{alg:Opt} above we have that for any
initial condition $x_0 \in \ball{\e}{\xf}$ Algorithm~\ref{alg:Opt} yields
a trajectory $x_k \to \xf$.
\end{lem}

\begin{figure} 
\begin{center}
\subfigure[]{\includegraphics[scale=1]{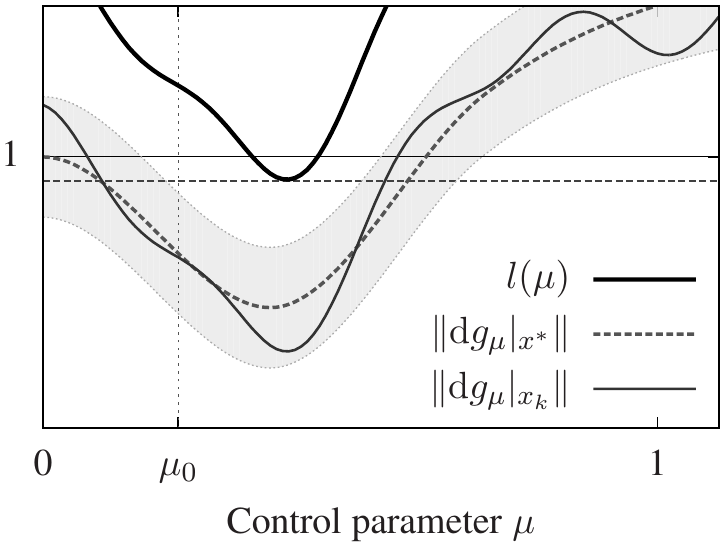}}\quad
\subfigure[]{\includegraphics[scale=1]{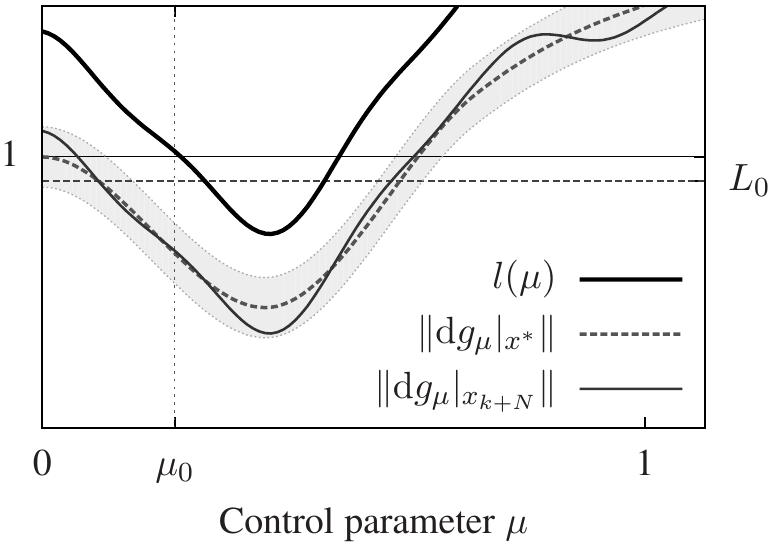}}
\end{center}
\caption{\label{fig:Optimization}(a) Inequality~\eref{eq:Opt} does not
always have to be satisfied. (b) Since $x_k\to \xf$ we have that after a
finite time $N$ the function $l(\mu)$ is below $L_0$ for some $\mu$.}
\end{figure}

\begin{proof}
If the optimization process does not give a result, convergence is ensured
by Proposition~\ref{prop:SD} and the Banach Fixed Point Theorem. Without
loss of generality, suppose that optimization yields a result for $k=1$.
Then because of~\eref{eq:Cond} and~\eref{eq:Opt} we have
\begin{align*}
\norm{\drv{g_{\mu_1}}{\xf}}_{\cball{\dist{x_1}{\xf}}{\xf}} &\leq \norm{\drv{g_{\mu_1}}{\xf}} + K\dist{x_1}{\xf}\\
&\leq \norm{\drv{g_{\mu_1}}{x_1}} + 2K\dist{x_1}{\xf}\\
&\leq \norm{\drv{g_{\mu_1}}{x_1}} + \left(\frac{2KL_0}{1-L_0}\right)\dist{x_1}{x_0}\\
& = L_1 \leq L_0 < 1.
\end{align*}
Therefore, $x_1$ is contained in a ball around the fixed point $\xf$ on
which the map $g_{\mu_1}$ is a contraction with contraction
coefficient $L_1$. The same calculation is valid for subsequent optimization
steps for $k>1$.
\pquad\end{proof}

The lemma above ensures that the adaptation does not compromise
convergence against the stabilized fixed point. But will optimization
actually take place? For a map with $K=0$, adaptation is not necessary
since $\norm{\drv{g_{\mu}}{\xf}}=\norm{\drv{g_{\mu}}{x_k}}$ and
therefore we can set $\mu$ straight to the optimal value.

\begin{lem}\label{lem:Opt}
With the assumptions of Algorithm~\ref{alg:Opt} Inequality~\eref{eq:Opt}
is satisfied every finitely many steps.
\end{lem}

\begin{proof}
By definition we have $\norm{\drv{g_{\mu_0}}{\xf}} \leq L_0 = \lambda_0
+\delta_0$ with $\delta_0>0$. Hence we have $\norm{\drv{g_{\mu_0}}{\xf}} < L_0$.
Let $\zeta > 0$ be such that $\zeta < L_0 - \norm{\drv{g_{\mu_0}}{\xf}}$.
Since $\dist{x_k}{x_{k+1}}$
is a Cauchy sequence and $\norm{\drv{g_{\mu_0}}{x_k}} \to
\norm{\drv{g_{\mu_0}}{\xf}}$ there is an $N \in \N$ such that
$\left(\frac{2KL_0}{1-L_0}\right)\dist{x_{N}}{x_{N-1}} < \frac{\zeta}{2}$
and $\abs{\norm{\drv{g_{\mu_0}}{\xf}} - \norm{\drv{g_{\mu_0}}{x_{N}}}} < \frac{\zeta}{2}$.
Thus, we have
\[\norm{\drv{g_{\mu_0}}{x_{N}}} + \left(\frac{2KL_0}{1-L_0}\right)\dist{x_{N}}{x_{N-1}}
< \norm{\drv{g_{\mu_0}}{\xf}} + \frac{\zeta}{2} + \frac{\zeta}{2}
< \norm{\drv{g_{\mu_0}}{\xf}} + \zeta < L_0.
\]
Therefore, Inequality~\eref{eq:Opt} will be satisfied after maximally $N=N_0$
steps.

Inductively, by increasing all indices above by $N$, the same argument
gives a sequence $N_l$, $l\in\N$, of indices for which Inequality~\eref{eq:Opt}
is satisfied. This completes the proof of the assertion.
\pquad\end{proof}

\begin{rem} 
Although the adaptation method gives a sequence $\mu_k$ that minimizes the
norm while ensuring convergence, it is not clear how often optimization
yields a result. Additional conditions on the map $f$, such as
requiring monotonicity of
$\norm{\drv{g_\mu}{x_k}}$ in $x_k$, influence how often the parameter
$\mu$ will be adapted. On the other hand, additional constraints make
the theory less broadly applicable.
\end{rem}

If Inequality~\eref{eq:Opt} is satisfied for some $k>0$, then, because
of continuity, it holds for a whole closed neighborhood of $\mu_k$.
This gives $\mu_{k+1}$ with $\norm{\drv{g_{\mu_{k+1}}}{\xf}} < 
\norm{\drv{g_{\mu_{k+1}}}{\xf}}$ unless $\norm{\drv{g_{\mu}}{\xf}}$ is
constant on that interval.

\begin{defn}[\cite{Householder1959}]
A matrix norm $\enorm$ on $\Rr^{n\times n}$ is called minimal
for $M\in\Rr^{n\times n}$ if $\rho(M)=\norm{M}$.
\end{defn}

The main results of this section can now be summarized in the
following theorem.

\begin{thm}\label{thm:Main}
Suppose $f\in\F(\mu_0, p)$ for $\mu_0>0$ such that $f$ satisfies~\eref{eq:Cond}
with $K \geq 0$ in a neighborhood of $\xf\in\sFP(f)$. Furthermore, let 
$\e, \lambda_0$, and $\delta_0$ be chosen as described above. Then
for any initial condition $x_0\in\ball{\e}{\xf}$ 
Algorithm~\ref{alg:Opt} minimizes an upper bound for the spectral
radius $\lda{\xf}{\mu}$.

In particular,
if the induced operator norm
$\|\ \cdot\ \|_{\textrm{op}}$ is minimal for $\drv{g_{\mu_\infty}}{\xf}$,
it converges at least linearly with asymptotic convergence
speed~$\lambda_\infty$.
\end{thm}

\begin{rem}
For dimension $n=1$, the Euclidean norm is minimal.
\end{rem}

\begin{proof}(of Theorem~\ref{thm:Main})
Lemmas~\ref{lem:Conv} and \ref{lem:Opt} ensure convergence against
the fixed point $\xf$ and adaptation of the control parameter $\mu$
after a maximum of some finite number of steps.

By construction, $\mu_k$ tends to a value which minimizes the norm of
the derivative of $g_\mu$ at $\xf$. For arbitrary dimension $n$ we have
$\lda{\xf}{\mu_k}\leq \norm{\drv{g_{\mu_k}}{\xf}}$. If in addition
the norm is minimal, in the limit the spectral radius is minimized
yielding
optimal asymptotic convergence speed, i.e., $\mu_k\to\mu_\infty$.
\pquad\end{proof}

\begin{rem}
One could also use convergence acceleration transformations
\cite{Smith1987} in order to get a better approximation to 
$\lda{\xf}{\mu}$. However, to exploit the acceleration within the
framework of this theory, one would have to have suitable error
estimates for the transformed sequence.
\end{rem}

The choice of the size neighborhood $\ball{\e}{\xf}$ depends on the
desired estimate of the contraction constant. It is clearly bounded
from above since we have to make sure that there is a contraction. On
the other hand, it is desirable to take a neighborhood as large as
possible to make the method applicable to as many initial conditions
as possible.

\subsection{From local to global}\label{sect:Global}

We want to consider the situation where the control is turned on at a
random point in time. We choose indices such that this time is $k=0$.
In general, the initial condition $x_0$ for the adaptation method is
unknown and it is likely to be outside of a neighborhood 
$\ball{\e}{\xf}$ as defined in the section above. One possible scenario
is the existence of a chaotic attractor that makes up part of the old
chaotic attractor and ending up in its basin of attraction when the
control parameter
is turned up. Another likely scenario is to have $x_0$ close to the
boundary of the basin of attraction of one of the stabilized fixed
points. An initial condition close to the basin boundary implies a
long transient iteration before the adaptation method becomes
applicable.

We want to quantify the latter scenario. Since we assumed the system
to follow the evolution equation $x_k=f(x_{k-1})$ for all $k\leq 0$,
$x_0$ is distributed on the attractor according to some $f$-invariant
measure $m$ on $A$. Suppose $m$ is an ergodic probability measure 
on~$A$. So $m(U)$ is the probability that $x_0 \in U$
at time $k=0$ for any measurable subset $U\subset A$.

Let $\ball{\e(\xf)}{\xf}$ denote the neighborhoods of $\xf\in\sFP(f)$
for which the acceleration method described in the previous section
is applicable for all initial conditions within that neighborhood.
Because of $f\in \F(\mu_0, p)$ at least one of these balls is not
empty. Define
\[V_0 = V = \left(\bigcup_{x\in \sFP(f)}\ball{\e(x)}{x}\right)\cap A\]
to be the part of the union of all these neighborhoods on the attractor.
Thus, if $V$ is measurable, $m(V)$ is a lower bound for the
probability that the adaptation method described above converges if the
parameter $\mu$ is set to $\mu_0$ at a random point in time and
dynamics evolved on the chaotic set before.
Furthermore, we define $V_k := \bigcup_{l\leq k}\ite{g_{\mu_0}}{(-l)}(V)$.
Now $P_k=m(V_k)$ is a lower bound on the probability that the algorithm
will converge after letting the transformed system evolve for $k$ time
steps after being initialized with $\mu=\mu_0$ at time $k=0$.

As $k \to \infty$ tends to infinity the set $V_k$ will converge to the union of the
basins of attraction of the stabilized fixed points. Hence, we obtain
a function \[\vphi(\mu_0) := \lim_{k\to\infty}m(V_k)\] 
depending on the initial parameter~$\mu_0$. The value 
$\liminf_{\mu_0\to 0}\vphi(\mu_0) = \alts{\vphi}$ for some
$\alts{\vphi}\in [0,1]$ determines the size of the basin of attraction of
the stabilized fixed point.


\section{Adaptive PFC for the Logistic Map}

As an example, we apply the PFC transformation to the logistic
map given by the quadratic polynomial $\ell_r(x)=rx(1-x)$ with the real
parameter $1<r \leq 4$. It is well known that there are parameter
values for which the dynamics is chaotic on some subset $A$ of the
unit interval $I=[0, 1]$. In particular, for $r=4$ the whole unit
interval is a chaotic set. Here, we study the period one orbits;
higher periods can be treated similarly.

\subsection{Calculating the Adaptation Parameters}

First, we want to calculate the quantities for the adaptation
method described
in \Sref{sect:Close}. The second derivative of $\ell_r$ exists
everywhere on $\Rr$ and is bounded on compact subsets. Within this section
let $h'$ denote
the derivative of a differentiable function $h$. Here, we treat the
cases $C_k\in\sset{\pm\id}$ simultaneously by allowing the control
parameter $\mu$ for the transformed function
\[g_{\mu, r}(x) = S(\mu, C_k)(\ell_r)(x) = x+\mu(\ell_r(x)-x)\]
to range within the interval $[-1, 1]$. For $\mu=1$ we obtain the
original system and around $\mu=0$ either, where $C_k=\id$ when
$\mu$ is positive and $C_k=-\id$
when $\mu$ is negative.

Since $\abs{g_{\mu, r}''(x)}=\abs{\mu}\abs{\ell_r''(x)}=2r\abs{\mu}$
for all $x\in I$,
the maximum in $\mu$ is taken for $\abs{\mu}=1$. Therefore, if we set
$K=8$, we obtain a constant independent also of the parameter~$r$ and
the sign of~$C_k$.

The two fixed points of $f_r$ are $\xf=0$ and $\xf=\frac{r-1}{r}$. The
derivatives at the fixed points are $g_{\mu, r}'(0)=1+\mu(r-1)$ and 
$g_{\mu, r}'(\frac{r-1}{r})=1-\mu(r-1)$. Hence, $\xf=0$ is stable for $\mu$
negative ($C_k=-\id$) and $\xf=\frac{r-1}{r}$ for $\mu$ positive
($C_k=\id$). To apply the adaptive method, the initial parameters
need to be determined as in \Sref{sect:Close}: for a given
$\mu_0$ the bound $\lambda_0$ can be calculated directly from the
derivative. Furthermore, we have to find $\e$ that defines a neighborhood
of $\xf$ for the initial condition $x_0$ and a given
initial $\mu_0$. From the local stability and~\eref{eq:Cond} we obtain
that convergence is ensured if
\begin{subequations}\label{eq:LoAll}
\begin{eqnarray}
\lambda_0 = 1+\abs{\mu_0}(r-1) &>-1\label{eq:Lo},\\
K\e-\abs{\mu_0}(r-1) &< 0,\\
K\e+\abs{\mu_0}(r-1) &< 2\label{eq:Lo2}
\end{eqnarray}
\end{subequations}
For either $\xf$ this gives $\abs{\mu_0}<\frac{2}{r-1}$. This results
in a bound for the size of the neighborhood of $\xf$ in which the map
is a contraction, \[\e<\min\left\{\frac{2-\mu_0(r-1)}{K}, 
\frac{\mu_0(r-1)}{K}\right\}.\] The optimal bound $\e<\frac{1}{K}$
is achieved for $\mu_0(r-1)=1$.

It is desirable to choose $\e$ as large as
possible (to cover as many initial conditions as possible) while
keeping the whole expression on the left hand side of~\eref{eq:Lo}
inequality as small as possible (a smaller contraction constant
$L_0$ leads to stronger contraction). This choice depends on the
initial guess $\mu_0$.

\begin{rem}
The chaotic set $A$ depends on the choice of the parameter $r$, so
we obtain a family of chaotic sets $A_r$. Note that we do not
necessarily have $0\in A_r$ or $\frac{r-1}{r} \in A_r$. We have
$A_4=I$ so that the two fixed points are contained in $A_4$.
Otherwise, for a given fixed point $\xf$, $\e$ has to be chosen
large enough such that $A_r\cap \ball{\e}{\xf} \neq \emptyset$.
\end{rem}

The constant parameters $K=8$ and $\mu_0, \lambda_0$, and $\e$ as
given by \eref{eq:LoAll} together with an
approximation of the measure on $A$ are the basis of the calculations of
lower bounds for the probability of convergence in the following
section.

\subsection{PFC on the complex plane}

The quadratic polynomial defining the logistic map can also be
seen as a polynomial over the complex numbers. Iteration of complex
polynomials is a classical example in complex analytic dynamics 
and the theory developed there can tell us something about the
effect of the PFC transformation $S(\mu, C_k)$ for $C_k\in\sset{\pm\id}$.
Here, this geometric point of view allows us to calculate the
full basin of attraction of the stabilized fixed points. In particular,
we also obtain convergence for general, complex valued initial
conditions in a neighborhood of the periodic orbits in the complex
plane.

Recall some notions from one-dimensional complex dynamics \cite{Milnor2006}.
Suppose $f:\C\to\C$ is holomorphic. A point $z\in\Ch = \C\cup\{\infty\}$
is said to be in the \emph{Fatou set $F(f)$} if there is an open
neighborhood of $z$ on which the family of iterates 
$\set{\ite{f}{k}}{k\in\N}$ is normal. Its
complement is called the \emph{Julia set $J(f)$} and constitutes
the boundary of all Fatou components that contain any stable periodic
orbits. Both of these sets are forward and backward invariant with
respect to the map $f$. The Julia set is a chaotic set in our sense.
Henceforth, we denote the complex variable by $z$.

Let $f\in\Cz$ be a complex polynomial. Note that the result of the
PFC transformation $g_{\mu, r}$ again is a polynomial of the same degree
in the complex variable $z$ unless $f$ is constant or $\mu=0$. The
logistic map is defined as a polynomial of degree two. The dynamics
of quadratic polynomials are conjugate to the dynamics of a polynomial
$z^2+c$ where $c$ is a complex parameter. The parameter $c$ can be
characterized in terms of the orbit of the only finite critical point
$z=0$ (since $(z^2+c)'(0)=0$), where the points for which that orbit
is bounded constitute the \emph{Mandelbrot set $\M$}. For $c\in\M$
there can be bounded Fatou components corresponding to the basin of
attraction of a stable periodic orbit.

The logistic family described above is conjugate to the subset 
$\left[-2, \frac{1}{4}\right]$ of the intersection of $\M$ with the
real axis. Since $g_{\mu, r}$ again
is a real quadratic polynomial of degree two for $\ell_r$ and these
polynomials only keep the real axis invariant if $c\in\Rr$, every
$g_{\mu, r}$ is conjugated to a quadratic polynomial $z^2+c$ with
real $c$. For given~$r$, the relationship between this complex
parameter $c$ and the control parameter $\mu$ is given by
\[c_r(\mu) = \frac{1}{4}\left(1-\mu^2(r-1)^2\right)\]
for $\mu\neq 0$. Hence, varying the parameter $\mu$ results in a
``shift'' of the dynamics up the real axis until it approaches
$c=\frac{1}{4}$ as $\mu\to 0$. From the equation above, one can
also see that the dynamics of $g_{\mu, r}$ are conjugated for $\mu=1$ and
$\mu=-1$, the former case corresponding to the unperturbed
system.

\begin{figure} 
\begin{center}
\includegraphics[scale=1]{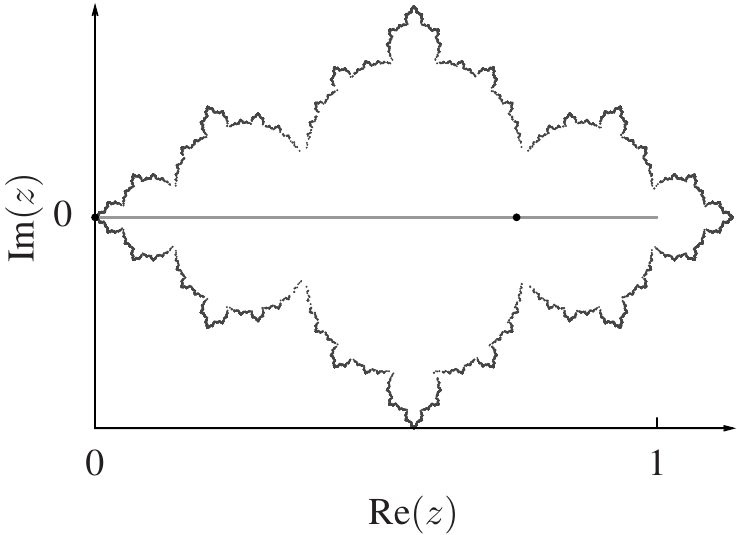}\qquad
\includegraphics[scale=1]{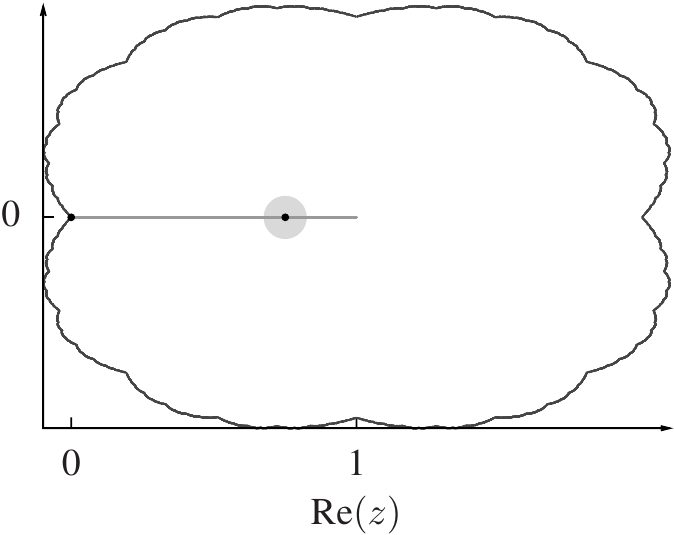}
\end{center}
\caption{\label{fig:JuliaId}Julia sets for parameters $\mu=-0.65$ (left)
and $\mu=-0.2$ (right) for $C_k=-\id$. The Julia set of $g_{\mu, 4}$ is depicted
in dark gray whereas the Julia set of the original map, i.e., the unit interval,
is depicted in medium gray. The fixed points $\zf=0$ and $\zf=\frac{3}{4}$
are marked by black dots. Shaded circles indicate $\ball{\e}{\zf}$ for
$\mu_0=\mu$.}
\end{figure}

What does stabilization of fixed points mean in terms of complex
analytic dynamics? An unstable fixed point is contained in the
Julia set. The goal of stabilization is to turn this fixed point
into a stable one, i.e., that it now belongs to a bounded Fatou
component. That is, the transformation should deform the
Julia set in a way that it does not contain the targeted
periodic point anymore.

Let us consider the case $r=4$ in more detail. The Julia set
$J(\ell_{r=4}) = I$ is equal to the whole unit interval. The probability
distribution $m$ is given by a beta distribution with parameters
both equal to $\frac{1}{2}$ (cf.~for example \cite{Diakonos1996}),
i.e., with probability density function
\[p(x)=\left(\pi x^{\frac{1}{2}}(1-x)^{\frac{1}{2}}\right)^{-1}.\]

Suppose $C_k=\id$ and $\mu$ small enough. In this case
$\zf=\frac{3}{4}$ is the stabilized fixed point. In the
previous section we calculated the maximum size of the ball around
the fixed point for which the adaptation method works straight
away. This radius is given by $\e<\frac{1}{8}$ and
\begin{align*}P_0 = m(V_0) &\leq m(\ball{\e}{\zf}\cap I) = m\left(\left[\frac{5}{8}, \frac{7}{8}\right]\right) = \int_{\frac{5}{8}}^{\frac{7}{8}}\frac{\ud x}{\pi x^{\frac{1}{2}}(1-x)^{\frac{1}{2}}}\\ &= \frac{2}{\pi}\left(\arctan\left(\sqrt{\frac{3}{5}}\right)-\arccot\left(\sqrt{7}\right)\right)\approx 0.1895.
\end{align*}
In \Fref{fig:JuliaId}, one can see that the whole unit
interval is contained in the bounded Fatou component. Backward
iteration takes this set closer to the boundary of this Fatou
component. This means that, for $\mu_0$ small,
\[\vphi(\mu_0) = 1,\]
$P_0 \leq P_k=m(V_k) \leq 1$ and $\alts{\vphi} = 1$.
Therefore, trajectories will converge to the stabilized
periodic point with probability one.

\begin{figure} 
\begin{center}
\includegraphics[scale=1]{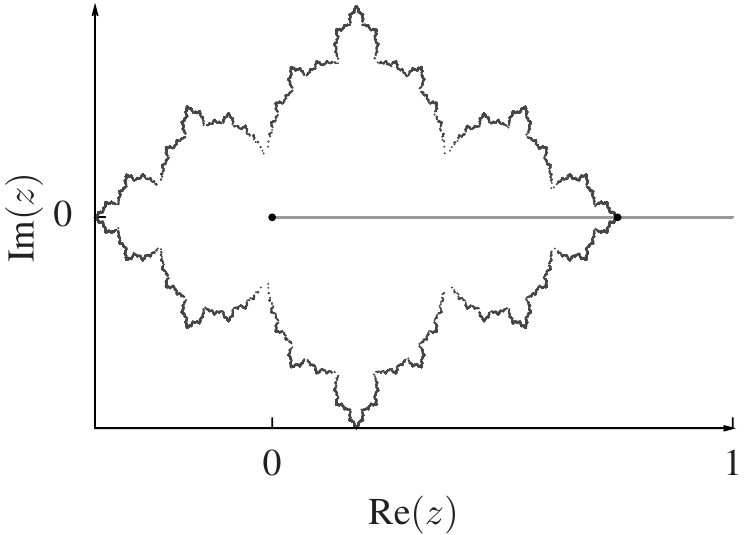}\qquad
\includegraphics[scale=1]{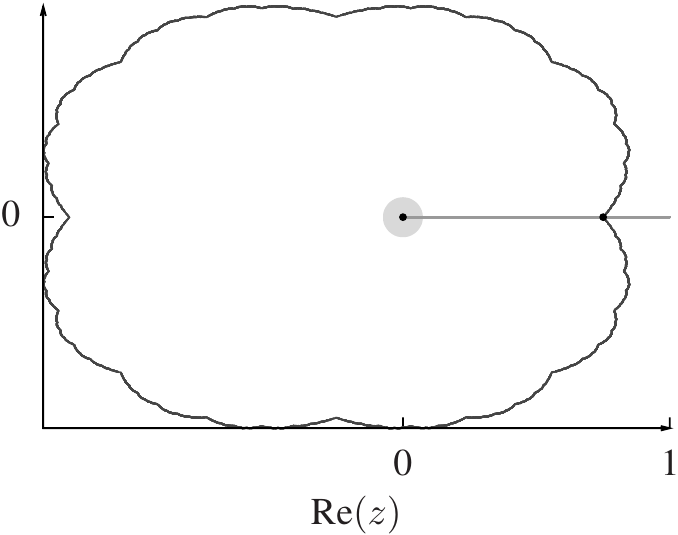}
\end{center}
\caption{\label{fig:JuliaMId}Julia sets for parameters $\mu=-0.65$ (left)
and $\mu=-0.2$ (right) for $C_k=\id$. The Julia set of $g_{\mu, 4}$ is depicted
in dark gray whereas the Julia set of the original map, i.e., the unit interval,
is depicted in medium gray. The fixed points $\zf=0$ and $\zf=\frac{3}{4}$
are marked by black dots. Shaded circles indicate $\ball{\e}{\zf}$ for
$\mu_0=\mu$.}
\end{figure}

The picture is slightly different for $C_k=-\id$ and $\mu$ small
enough. Now, $\zf=0$ is the stabilized fixed point. Again
we have $\e<\frac{1}{8}$ and therefore
\begin{align*}P_0=m(V_0) &\leq m(\ball{\e}{\zf}\cap I) = m\left(\left[0, \frac{1}{8}\right]\right) = \int_{0}^{\frac{1}{8}}\frac{\ud x}{\pi x^{\frac{1}{2}}(1-x)^{\frac{1}{2}}}\\ &= \frac{2}{\pi}\arccot\left(\sqrt{7}\right)\approx 0.2301.
\end{align*}
In this case backward iteration yields a different result as can
be seen in \Fref{fig:JuliaMId}. Part of the set of initial
conditions $A$ is in the basin of attraction of infinity and the
intersection with the Julia set is exactly the fixed point
$\zf=\frac{3}{4}$. Therefore, the probability to converge
with a random initial condition is less than one. Integrating the
probability density function gives 
\[\vphi(\mu_0) = \int_{0}^{\frac{3}{4}}\frac{\ud x}{\pi x^{\frac{1}{2}}(1-x)^{\frac{1}{2}}} = \frac{2}{3}.\]
Therefore, we have $\alts{\vphi} = \frac{2}{3}$ and $P_0 \leq m(V_k)
\leq\frac{2}{3}$. In contrast to the case $C_k=\id$ this means that
for $C_k=\id$ and $\zf=0$ a trajectory with an initial condition on
$I$ distributed according to $m$ will diverge with a probability of
one third.

When considering higher periods of such a polynomial map, the
Julia sets are more complicated as the degree of the iterated
polynomial rises exponentially. The situation changes qualitatively
when considering the predictive feedback control dynamics of higher
dimensional maps by interpreting them as functions $\Cn\to\Cn$.
In general, the dynamics of holomorphic, higher-dimensional maps
is more diverse since even low-dimensional invertible maps give
rise to complicated dynamics \cite{Hubbard1994}.


\section{Numerical Results}

To compare the speed of the adaptive method (ACC) with the original
PFC chaos control in a real world application, 
we performed numerical
simulations for the logistic map~$\ell_4$. The results for $C_k=+\id$,
$\mu_0\in[0, 1]$ and periods one and two are summarized in 
\Fref{fig:Numerics}. One can clearly see that for most initial
values of the control parameter, the adaptive method yields an
increase in convergence speed. The results for $C_k=-\id$ and period
one are similar but the convergence probability is lower (not shown)
in accordance with the results of the previous section, 
cf.~\Fref{fig:JuliaMId}. In the case of period $p=2$, the orbits
stabilized by negative~$\mu$ are the period one orbits, which is
reflected in
our numerical results (not shown). A non-optimized, ad hoc choice of
parameters for the adaptive method of $K=8$ and $L_0=0.99$ (independent
of the initial condition) were employed in the simulations. The criterion
for convergence time~$T$ was given by $\abs{x_{T} - x_{T-1}}\le 10^{-10}$ but
reliability was determined after checking for the correct period.

\begin{figure}
\begin{center}
\includegraphics[scale=1]{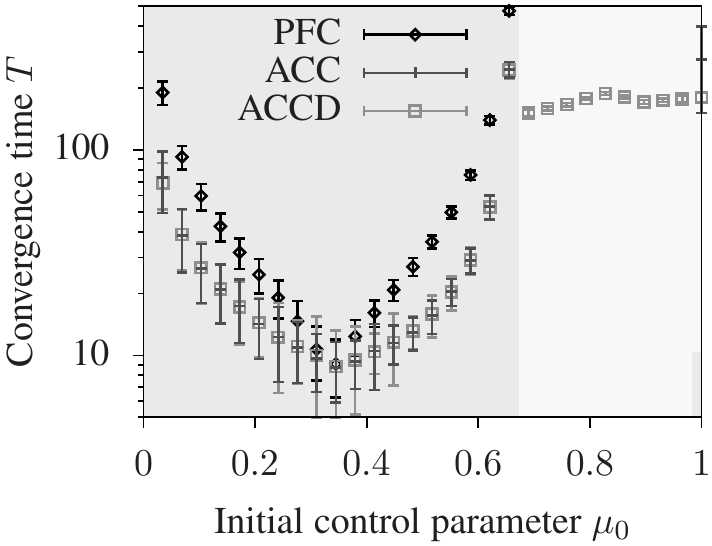}\quad
\includegraphics[scale=1]{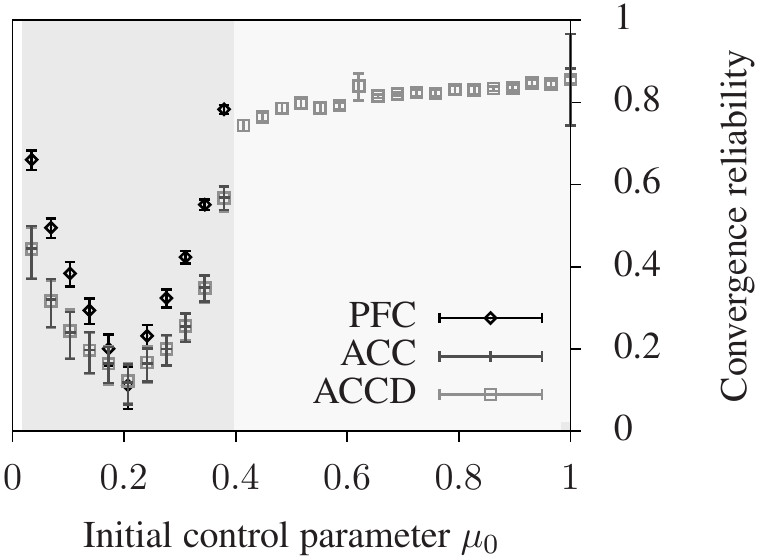}
\end{center}
\caption{\label{fig:Numerics}Speed and reliability comparison of
the original PFC chaos control ($\mu=\mu_0$ fixed) and both ACC and 
ACCD for the logistic map~$\ell_4$ with $p=1$ (left) and $p=2$ (right).
Times are only plotted if more than $1\%$ of initial values that
lead to convergence to the correct period. Gray shading indicate the
convergence reliability; dark gray corresponds to all methods
converging, light gray to the reliability of ACCD. Convergence
time~$T$ is given by $\abs{x_{T} - x_{T-1}}\le 10^{-10}$ calculated for
$1000$ random initial conditions after a transient of random length.}
\end{figure}

The convergence reliability, i.e., the fraction of trials where the
above criterion is fulfilled after some time $T$, is not improved by
the adaptive method. However, it is possible to amend the adaptation
method to lead to convergence for most initial
conditions~$x_0$ within the convergent regime, independent of the
initial value of the control parameter~$\mu_{0}$ (the modified
method is denoted by ACCD). When adapting,
the ACC method has to check whether Criterion~\eref{eq:Opt}
is fulfilled. If this is not the case after $M$ iterations, the
modified method simply reduces~$\mu$ to a certain fraction~$v$.
To prevent~$\mu$ of becoming too small, we imply a threshold~$\theta$
below which~$\mu$ cannot decay. Put in other words, the modified
method ACCD will automatically decrease $\mu$ towards zero to
reach the
convergence regime if Inequality~\eref{eq:Opt} is not satisfied
within a given number of steps.

The modified method ACCD behaves like the original ACC method
for initial values of $\mu_0$ in the convergent regime while leading
to convergence outside of it, cf. \Fref{fig:Numerics} (here
$M=50$, $v=0.7$, $\theta=0.1$ for period $p=1$, and $\theta=0.05$
for period~$p=2$). Failure of convergence that is due to
the existence of a range of diverging initial conditions, however,
will persist, even with the decay. The results are similar for a
broad parameter range (e.g., decay rate \mbox{$v\in[0.65, 0.99]$} and
decay kick-in time \mbox{$M\in[10, 100]$}). For a decay rate too
close to $100\%$ or a too large decay kick-in time,
it will take many iterations to reach the convergent interval. On
the other hand, if the decay kick-in time is too small or does not
exist at all, $\mu_k$ decreases even if it is in the convergent
interval as Criterion~\eref{eq:Opt} is not fulfilled all the time,
unnecessarily increasing convergence time.

\section{Discussion}

Here we presented a method which adapts the control parameter
$\mu$ of the PFC method in order to accelerate the systems convergence
to a periodic orbit. In contrast to heuristically chosen methods,
our adaptation does not compromise convergence. The algorithm 
converges in a
neighborhood of every periodic orbit that was stabilized by the
stabilizing transformation. Assuming the existence of an
invariant, ergodic probability measure on the chaotic attractor, we obtain
an analytic bound for the probability $P_k$ that the system
converges to a periodic orbit if
the chaos control is switched on at an arbitrary point in time.
Although these results are stated in the framework of discrete time
dynamical systems, they can also be applied to stabilize continuous
time systems after discretization such as taking Poincar\'e
sections. The logistic map provides an example for which we can
calculate the parameters for the method. We estimated the
probability of convergence and highlighted its dependence on the
fixed point to be stabilized and the associated matrix~$C_k$.

Our method was stated in the general context of ``chaotic
sets.'' In general, such sets do not need to be local or even
global attractors of the dynamical system. In fact, the Julia sets
considered in the example are repelling rather than attracting. In
applications,
however, an attractor would be desirable such that the process of
stabilization becomes repeatable. That is, after the control
perturbation is turned off by choosing the appropriate value for the
control parameter, the dynamics would return to the attractor
and the process can be started over again.

Apart from its role as a chaos control method, the estimates
described in \Sref{sect:Global} give information on the
PFC method itself. It allowed us to calculate the size of the
basin of attraction for varying $\mu$ in our example. Decreasing
$\mu$ always leads to slower convergence since the eigenvalues
converge to one
as \mbox{$\mu \to 0$}. So is it possible to find an optimal
$\mu_0$ for a given map? Since any adaptation method increases
the computational cost of the chaos control method, \emph{a priori}
estimates of such crucial quantities are of importance. Furthermore,
the choice of
the stabilization matrix $C_k$ depends on the type of fixed points in
the chaotic attractor. Hence, global statistics for a given map $f$ of
the periodic orbits and their stability properties might yield some
\emph{a priori} estimates.

Our numerical studies suggest that it is possible to get reliable
convergence without \emph{a priori} knowledge of the exact values
for the parameters. A slight modification of the method yields
a hybrid method that finds the regime of control parameter
in which the dynamics converge online and then adapt it to the
optimal value. This simplification, however, comes at a cost in
convergence
speed. By definition, PFC cannot distinguish between a periodic
orbit of period $p$ and any $q|p$, a divisor of $p$. Our numerical
calculations, however, indicate that this does not influence
reliability of the chaos control method. This is most likely
caused by the exponential growth of the number of periodic orbits.
In the future, it would be desirable to add a mechanism that 
rigorously distinguishes between the target period and its divisors
to prove optimal convergence.

An adaptation method for chaos control is a step towards solving
the intuitively contradictory problem of optimizing speed while
maintaining simplicity in implementation. However, as discussed
above, it leads to further challenging research questions that
have to be addressed in the future.

\section*{Acknowledgements}

The authors would like to thank 
Eckehard Sch\"oll for helpful discussions and
Christoph Kirst for valuable comments on
the manuscript.
CB would like to thank Laurent Bartholdi for
making this project possible.
This work was supported by the Federal
Ministry of Education and Research (BMBF) by grant numbers 01GQ1005A
and 01GQ1005B.


\bibliographystyle{unsrt}
\bibliography{ChaosControl}

\end{document}